\documentclass{article}
\usepackage[utf8]{inputenc}
\usepackage{lmodern}

\usepackage{graphicx} 
\usepackage{amsmath}
\usepackage{amssymb}
\usepackage{amsthm}

\usepackage{geometry}
\usepackage{microtype}

\usepackage{authblk}
\usepackage{cleveref}

\newtheorem{lemma}{Lemma}
\newtheorem{theorem}{Theorem}

\title{On the Inapproximability of Finding Minimum Monitoring Edge-Geodetic Sets}
\author{\parbox{7cm}{\centering Davide Bilò \authorcr \texttt{davide.bilo@univaq.it}}}
\author{\parbox{7cm}{\centering Giordano Colli \authorcr \texttt{giordano.colli@student.univaq.it}} \authorcr\medskip} 
\author{\parbox{7cm}{\centering Luca Forlizzi \authorcr\texttt{luca.forlizzi@univaq.it}}}
\author{\parbox{7cm}{\centering Stefano Leucci \authorcr \texttt{stefano.leucci@univaq.it}}}

\affil{Department of Information Engineering, Computer Science and Mathematics, University of L'Aquila, Italy}

\usepackage{xspace}
\newcommand{\p}{\ensuremath{\mathsf{P}}\xspace}
\newcommand{\np}{\ensuremath{\mathsf{NP}}\xspace}

\begin{document}

\maketitle

\begin{abstract}
    Given an undirected connected graph $G = (V(G), E(G))$ on $n$ vertices, the minimum \textsc{Monitoring Edge-Geodetic Set} (MEG-set) problem asks to find a subset $M \subseteq V(G)$ of minimum cardinality such that, for every edge $e \in E(G)$, there exist $x,y \in M$ for which all shortest paths between $x$ and $y$ in $G$ traverse $e$.

    We show that, for any constant $c < \frac{1}{2}$, no polynomial-time $(c \log n)$-approximation algorithm for the minimum MEG-set problem exists, unless $\p = \np$.
\end{abstract}

\section{Introduction}
    We study the minimum \textsc{Monitoring Edge-Geodetic Set} (MEG-set) problem. 
    Given an undirected and connected graph $G = (V(G), E(G))$ on $n$ vertices,
    we say that an edge $e \in E(G)$ is monitored by a pair $\{x,y\}$ if $e$ lies on  \emph{every} shortest path between $x$ and $y$ in $G$. Moreover, we say that a set $M \subseteq V(G)$ monitors $e$ if there exist $x,y \in M$ such that  $\{x,y\}$ monitors $e$.
    A MEG-set of $G$ is a subset $M \subset V(G)$ that monitors all edges in $E(G)$.
    The goal of the minimum MEG-set problem is that of finding a MEG-set of $G$ of minimum cardinality.
    This problem has been introduced in \cite{Monitoring_edge-geodetic_sets_in_graphs}, where the authors focus on providing upper and lower bounds to the size of minimum MEG-sets for both general graphs and special classes of graphs. Further bounds on the size of MEG-sets have been given in \cite{haslegrave, FoucaudMMSST24, TanLWL23, MaJYL24, XuYBZS24}.
    
    The minimum MEG-set problem was proven to be \np-hard by \cite{haslegrave} on general graphs, and by \cite{FoucaudMMSST24} on $3$-degenerate $2$-apex graphs. 
    However, to the best of our knowledge, no inapproximability result is currently known.

    In this paper, we show that the problem of finding a MEG-set of minimum size is not approximable within a factor of $c \ln n$, where $c < \frac{1}{2}$ is a constant of choice, unless $\p = \np$. 
    
\section{Preliminaries}

\begin{lemma}
    \label{lemma:leaf}
    Let $v$ be a vertex of degree $1$ in $G$. Vertex $v$ belongs to all MEG-sets of $G$.
\end{lemma}
\begin{proof}
    Assume towards a contradiction that some MEG-set $M \subseteq V(G) \setminus \{v\}$ of $G$ exists.
    Let $\{x, y\}$ with $x,y \in M$ be a pair of vertices that monitors the unique edge incident to $v$.
    Since $\{x,y\} \cap \{v\} = \emptyset$, all shortest path from $x$ to $y$ in $G$ contain $v$ as an internal vertex, thus $v$ must have at least two incident edges, contradicting the hypothesis of the lemma.
\end{proof}

\begin{lemma}
    \label{lemma:leaf_neighbor}
    Let $u$ be a vertex of degree $1$ in $G$ and let $v$ be its sole neighbor. If $|V(G)| \ge 3$ and $M$ is a MEG-set of $G$ then $M \setminus \{v\}$ is a MEG-set of $G$.
\end{lemma}
\begin{proof}
    From \Cref{lemma:leaf} we know that $u \in M$.
    We start by observing that $\{u,v\}$ only monitors edge $(u,v)$ and, since $|E(G)| \ge |V(G)|-1 \ge 2$, there must exist some vertex $x \in M \setminus \{u,v\}$.
    
    Since $(u,v)$ is a bridge of $G$, it is traversed by all paths between $u$ and any vertex in $V(G) \setminus \{u\}$ and, in particular, it is monitored by $\{u, x\}$.

    We now turn our attention to the edges in $E \setminus \{ (v,u) \}$ and show that any such edge $e$ that is monitored by $\{v,y\}$, with $y \in M$, is also monitored by $\{u, y\}$. Notice, since $e \neq (u,v)$, we have $y \neq u$.
    Consider any shortest path $P$ from $u$ to $y$ in $G$ and observe that $P$ consists of the edge $(u,v)$ followed by a path $P'$ from $v$ to $y$. By suboptimality of shortest paths, $P'$ is a shortest path from $v$ to $y$ in $G$. Since $\{v, x\}$ monitors $e$, $P'$ contains $e$ and so does $P$.  Thus, $\{u, y\}$ monitors $e$.
\end{proof}

\section{Our Inapproximability Result}

We reduce from the \textsc{Set Cover} problem. A \textsc{Set Cover} instance $\mathcal{I}=\langle X, \mathcal{S}\rangle$ is described as a set of $\eta$ items $X = \{x_1, \dots, x_\eta\}$, and a collection $\mathcal{S} = \{ S_1, \dots, S_h\}$ of $h \ge 2$ distinct subsets of $X$, such that each subset contains at least two items and each item appears in at least two subsets.\footnote{
This can be guaranteed w.l.o.g.\ by repeatedly reducing the instance by applying the first applicable of the following two reduction rules.
Rule 1: if there exists an item $x_i$ that is contained by a single subset $S_j$, then  $S_j$ belongs to all feasible solutions, and we reduce to the instance in which both $S_j$ and $x_i$ have been removed. 
Rule 2: if there exists a subset $S_j$ that contains a single element, then (due to Rule 1) there is an optimal solution that does not contain $S_j$, and we reduce to the instance in which $S_j$ has been removed.
Notice that this process can only decrease the values of $\eta$ and $h$.} The goal is that of computing a collection $\mathcal{S}^* \subseteq \mathcal{S}$ of minimum size such that $\cup_{S_i \in \mathcal{S}^*} S_i = X$.\footnote{We assume w.l.o.g.\ that $\cup_{i=1}^h S_i = X$, i.e., that a solution exists.}

It is known that, unless $\p \neq NP$, the \textsc{Set Cover} problem is not approximable within a factor of $(1-\varepsilon) \ln |\mathcal{I}|$, where $\varepsilon > 0$ is a constant and $|\mathcal{I}|$ is the size of the \textsc{Set Cover} instance \cite{DinurS14}.

Given an instance $\mathcal{I} = \langle X, \mathcal{S} \rangle$ of \textsc{Set Cover}, we can build an associated bipartite graph $H$ whose vertex set $V(H)$ is $X \cup \mathcal{S}$ and such that $H$ contains edge $(x_i, S_j)$ if and only if $x_i \in S_j$.  We define $N=h+\eta$. Observe that $|\mathcal{I}| \geq N$.

Let $k$ be an integer parameter, whose exact value will be chosen later, that satisfies $2 \le k = O(\text{poly}(N))$. We construct a graph $G$ that contains $k$ copies $H_1,\ldots,H_k$ of $H$ as induced subgraphs. In the following, for any $\ell=1,\ldots,k$, we denote by $x_{i,\ell}$ and $S_{j,\ell}$ the vertices of $H_\ell$ corresponding to the vertices $x_i$ and $S_j$ of $H$, respectively. More precisely, we build $G$ by starting with a graph that contains exactly the $k$ copies $H_1,\ldots,H_k$ of $H$ and augmenting it as follows (see \Cref{fig:reduction}):
\begin{itemize}
    \item For each item $x_i \in X$, we add two new vertices $y_i$, $y'_i$ along with the edge $(y_i, y'_i)$ and all the edges in $\{(x_{i,\ell}, y_i) \mid \ell=1,\ldots, k\}$;
    \item We add all the edges $(y_i, y_{i'})$ for all $1 \le i \le i' \le \eta$, so that the subgraph induced by $y_1, \dots, y_\eta$ is complete.
    \item For each set $S_j \in \mathcal{S}$, we add two new vertices $z_j$, $z'_j$ along with the edge $(z_j, z'_j)$, and all the edges in $\{(S_{j,\ell}, z_j) \mid \ell=1,\ldots, k\}$.
\end{itemize}
Observe that the number $n$ of vertices of $G$ satisfies $n= 2h + 2\eta + k(\eta + h) = (k+2)(\eta + h) = (k+2)N$.

\begin{figure}
    \centering
    \includegraphics{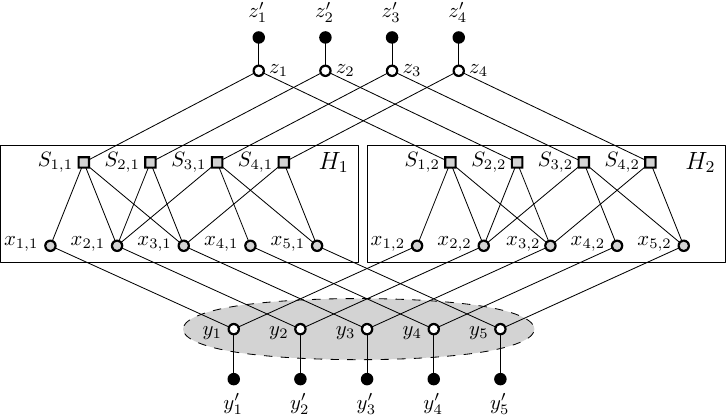}
    \caption{The graph $G$ obtained by applying our reduction with $k=2$ to the \textsc{Set Cover} instance $\mathcal{I} = \langle X, \mathcal{S}\rangle$ with $\eta=5$, $h=4$, $S_1 = \{x_1, x_2, x_3\}$, $S_2 = \{x_2, x_3\}$, $S_3 = \{x_2, x_4, x_5\}$,  and $S_4 = \{x_3, x_5\}$. To reduce clutter, the edges of the clique induced by the vertices $y_i$ (in the gray area) are not shown.}
    \label{fig:reduction}
\end{figure}

Let $Y = \{y_i \mid  i=1, \dots, \eta \}$, $Y' = \{y'_i \mid  i=1, \dots, \eta \}$, $Z = \{z_j \mid  i=1, \dots, h \}$, and
$Z' = \{z_j \mid  i=1, \dots, h \}$.
Moreover, define $L$ as the set of all vertices of degree $1$ in $G$, i.e., $L = Y' \cup Z'$. By \Cref{lemma:leaf}, the vertices in $L$ belong to all MEG-sets of $G$.

\begin{lemma}
    \label{lemma:edges_in_G_minus_H_monitored}
    $L$ monitors all edges having both endvertices $Y \cup Y' \cup Z \cup Z'$.
\end{lemma}
\begin{proof}
    Observe that all shortest paths from $y'_i \in L$ (resp.\ $z'_j \in L$) to any other vertex  $v \in L \setminus \{y'_i\}$ (resp.\ $v \in L \setminus \{ z'_j \}$) must traverse the sole edge incident to $y'_i$ (resp.\ $z'_i$), namely $(y'_i, y_i)$ (resp.\ $z'_j, z_j$).
    Since $|L| \ge 2$, such a $v$ always exists, and all edges incident to $L$ are monitored by $L$.

    The only remaining edges are those with both endpoints in $Y$. Let $(y_i, y_{i'})$ be such an edge. Since the only shortest path between $y'_i$ and $y'_{i'}$ in $G$ is $\langle y'_i, y_i, y_{i'}, y'_{i'} \rangle$, the pair $\{ y'_i, y'_{i'} \}$ monitors $(y_i, y_{i'})$.
\end{proof}

\begin{lemma}\label{lemma:set_cover_induces_meg_set}
    Let $\mathcal{S}_1,\ldots,\mathcal{S}_k$ be $k$ (not necessarily distinct) set covers of $\mathcal{I}$. The set $M = L \cup \{S_{j,\ell} \mid S_j \in \mathcal{S}_\ell, 1 \leq \ell \leq k\}$ is a MEG-set of $G$. 
\end{lemma}
\begin{proof}
    Since $L \subseteq M$, by \Cref{lemma:edges_in_G_minus_H_monitored}, we only need to argue about edges with at least one endvertex in some $H_\ell$, with $1\leq \ell\leq k$.
    Let $S_j \in \mathcal{S}_\ell$, and consider any $x_i \in S_j$. Edge $(S_{j,\ell},z_j)$ is monitored by $\{z'_j, S_{j,\ell}\}$. Edges $(S_{j,\ell}, x_{i,\ell})$ and $(x_{i,\ell}, y_i)$ are monitored by $\{S_{j,\ell}, y'_i\}$.

    The only remaining edges with at least one endvertex in $H_\ell$ are those incident to vertices $S_{j,\ell}$ with $S_j \in \mathcal{S} \setminus \mathcal{S}_\ell$. Consider any such $S_j$, let $x_i$ be an item in $S_j$, and let $S_k \in \mathcal{S}_\ell$ be any set such that $x_i \in S_k$ (notice that both $x_i$ and $S_k$ exist since sets are non-empty and $\mathcal{S}_\ell$ is a set cover).
    Edge $(S_{j,\ell}, x_{i,\ell})$ is monitored by $\{S_{k,\ell}, z'_j\}$, which also monitors $(S_{j,\ell}, z_j)$. 
\end{proof}

We say that a MEG-set $M$ is \emph{minimal} if, for every $v \in M$, $M \setminus \{v\}$ is not a MEG-set. \Cref{lemma:leaf_neighbor} ensures that any minimal MEG-set $M$ does not contain any of the vertices $y_i$, for $i=1,\dots,\eta$, or $z_j$ for $j=1,\dots,h$.
Hence, $M \setminus L$ contains only vertices in $\bigcup_{\ell=1}^k V(H_\ell)$.

\begin{lemma}\label{lemma:what_monitors_minimal_meg_set}
    Let $M$ be a minimal MEG-set of $G$. For every $i=1,\ldots, \eta$ and every $\ell=1,\ldots,k$, $M$ contains at least one among $x_{i,\ell}$ and all $S_{j,\ell}$ such that $x_i \in S_j$.
\end{lemma}
\begin{proof}
    Let $S(x_{i,\ell})$ be the set of all $S_{j,\ell}$ such that $x_i \in S_j$.  
    We show that no pair $\{u,v\}$ with $u,v \in M \setminus ( \{x_{i,\ell}\} \cup S(x_{i,\ell}) )$ can monitor the edge $(x_{i,\ell}, y_i)$. 

    Notice that, since $M$ is minimal, $M$ does not contain any vertex $y_{i'}$, for $i'=1, \dots, \eta$, nor any vertex $z_{j'}$ for $j'=1, \dots, h$.

    Let $\lambda \in \{1, \dots, h\} \setminus \{ \ell \}$ (such a $\lambda$ always exists since $h \ge 2$), and notice that any path $P$ in $G$ that has both endvertices in $V(G) \setminus V(H_\ell)$  can always be transformed into a walk $P'$ with as many edges as $P$ by replacing any occurrence of a vertex from copy $H_\ell$ of $H$ with the occurrence of the corresponding vertex from copy $H_{\lambda}$.  Since $P'$ does not contain $x_{i, \ell}$, this implies that no pair $\{u,v\}$ with $u,v \in V(G) \setminus V(H_\ell)$ can monitor $(x_{i,\ell}, y_i)$.

    The above discussion allows us to restrict ourselves to the situation in which at least one vertex of $\{u,v\}$ belongs to $V(H_\ell)$. 
    We sat that a vertex of $G$ is an \emph{item-vertex} if it is some $x_{i', \ell'}$ with $i'=1,\dots, \eta$ and $\ell' = 1, \dots, k$. Analogously, a vertex of $G$ is a \emph{set-vertex} if it is some $S_{j', \ell'}$ with $j'=1,\dots, h$ and $\ell' = 1, \dots, k$.
    We consider the following cases and, for each of them we exhibit a shortest path $\widetilde{P}$ from $u$ to $v$ in $G$ that does not traverse $(x_{i,\ell}, y_i)$:

    \begin{description}
        \item[Case 1: $u$ is an item-vertex in $H_\ell$ and $v$ is a set-vertex.] In this case $u = x_{i', \ell}$ for some $i' \neq i$, and we distinguish the following sub-cases:
        \begin{itemize}
            \item If $v = S_{j', \ell}$ and $x_{i'} \in S_{j'}$, then $\widetilde{P}$ consists of the edge $(x_{i', \ell}, S_{j', \ell})$.

            \item If $v = S_{j', \ell}$, $x_{i'} \not\in S_{j'}$, and there exists some 
            $S_{j''}$ such that $x_{i'} \in S_{j''}$ and $S_{j'} \cap S_{j''} \neq \emptyset $, then let $x_{i''} \in S_{j'} \cap S_{j''}$. 
            We choose $\widetilde{P} = \langle  x_{i', \ell}, S_{j'', \ell}, x_{i'', \ell}, S_{j', \ell} \rangle$.
            
            \item If $v = S_{j', \ell}$, $x_{i'} \not\in S_{j'}$, and there exists no 
            $S_{j''}$ such that $x_{i'} \in S_{j''}$ and $S_{j'} \cap S_{j''} \neq \emptyset$, then   
            let $x_{i''} \in S_{j'}$ (notice that $x_i \not\in S_{j'}$), and choose  $\widetilde{P} = \langle x_{i', \ell}, y_{i'}, y_{i''}, x_{i'', \ell}, S_{j', \ell} \rangle$.
            
            \item If $v = S_{j', \ell'}$, $x_{i'} \in S_{j'}$, and $\ell' \neq \ell$, then choose $\widetilde{P} = \langle x_{i', \ell}, S_{j', \ell}, z_{j'}, S_{j', \ell'}  \rangle$.
            
            \item If $v = S_{j', \ell'}$, $x_{i'} \not\in S_{j'}$, and $\ell' \neq \ell$, then let $x_{i''} \in S_{j'}$ (notice that $x_i \not\in S_{j'}$) and choose $\widetilde{P} = \langle  x_{i', \ell}, y_{i'}, y_{i''}, x_{i'', \ell'}, S_{j', \ell'} \rangle$.
         \end{itemize}
         
         \item[Case 2: $u$ is an item-vertex in $H_\ell$ and $v$ is an item-vertex.]
          In this case $u = x_{i', \ell}$ for some $i' \neq i$, and we distinguish the following sub-cases:
         \begin{itemize}
             \item If $v = x_{i'', \ell}$ with $i'' \neq i$ and there exists some set $S_{j'}$ such that $x_{i'}, x_{i''} \in S_{j'}$, then choose $\widetilde{P} = \langle  x_{i', \ell}, S_{j', \ell}, x_{i'', \ell} \rangle$.
             \item If $v=x_{i'', \ell}$ with $i'' \neq i$ and there is no set  $S_{j'}$ such that $x_{i'}, x_{i''} \in S_{j'}$, then choose $\widetilde{P} = \langle  x_{i', \ell}, y_{i'}, y_{i''}, x_{i'', \ell} \rangle$.
             \item If $v=x_{i', \ell''}$ with $\ell'' \neq \ell$, then choose $\widetilde{P} = \langle  x_{i', \ell}, y_{i'}, x_{i', \ell''} \rangle$.
             \item If $v = x_{i'', \ell''}$ with $i'' \neq i'$ and $\ell'' \neq \ell$, then choose $\widetilde{P} = \langle x_{i', \ell}, y_{i'}, y_{i''}, x_{i'', \ell''} \rangle$.  
         \end{itemize}
         
        \item[Case 3: $u$ is a set-vertex in $H_\ell$ and $v$ is a set-vertex.] In this case $u = S_{j', \ell}$ for some $j'$ such that $x_{i, \ell} \not\in  S_{j'}$, and we distinguish the following sub-cases:
        \begin{itemize}
            \item If $v = S_{j'', \ell}$ and there exists some $x_{i'}$ such that $x_{i'} \in S_{j'} \cap S_{j''}$, then $i' \neq i$ and we pick $\widetilde{P} = \langle S_{j', \ell}, x_{i', \ell}, S_{j'', \ell} \rangle$. 
            
            \item If $v = S_{j'', \ell}$, $S_{j'} \cap S_{j''} = \emptyset$, and there exists some $S_{j'''}$ such that $S_{j'} \cap S_{j'''} \neq \emptyset$ and  $S_{j''} \cap S_{j'''} \neq \emptyset$ then choose any 
            $x_{i'} \in S_{j'} \cap S_{j'''}$ and any $x_{i''} \in S_{j''} \cap S_{j'''}$. We pick $\widetilde{P} = \langle S_{j', \ell}, x_{i', \ell}, S_{j''', \ell}, x_{i'', \ell}, S_{j'', \ell} \rangle$.
            
            \item If $v = S_{j'', \ell}$, $S_{j'} \cap S_{j''} = \emptyset$, and there is no $S_{j'''}$ such that $S_{j'} \cap S_{j'''} \neq \emptyset$ and  $S_{j''} \cap S_{j'''} \neq \emptyset$, then choose $x_{i'} \in S_{j'} \setminus \{x_i\}$, $x_{i''} \in S_{j''} \setminus \{x_i\}$, and pick
            $\widetilde{P} = \langle S_{j', \ell}, x_{i', \ell}, y_{i'}, y_{i''} , x_{i'', \ell}, S_{j'', \ell} \rangle$.
            
            \item If $v = S_{j', \ell'}$ with $\ell' \neq \ell$, then we pick $\widetilde{P} = \langle S_{j', \ell}, z_{j'}, S_{j', \ell'} \rangle$. 

            \item If $v = S_{j'', \ell'}$ with $j'' \neq j'$ and $\ell' \neq \ell$, and there exists some item $x_{i'} \in S_{j'} \cap S_{j''}$, then we choose $\widetilde{P} = \langle S_{j', \ell}, z_{j'}, S_{j', \ell'}, x_{i', \ell'}, S_{j'', \ell'} \rangle$. 
            
            \item If $v = S_{j'', \ell'}$ with $j'' \neq j'$ and $\ell' \neq \ell$, and $S_{j'} \cap S_{j''} = \emptyset$, then let $x_{i'} \in S_{j'} \setminus \{x_i\}$, $x_{i''} \in S_{j''} \setminus \{x_i\}$. We pick $\widetilde{P} = \langle S_{j', \ell}, x_{i', \ell}, y_{i'}, y_{i''}, x_{i'', \ell'}, S_{j'', \ell'} \rangle$. 
        \end{itemize}
        
        \item[Case 4: $u$ is an item-vertex in $H_\ell$ and $v \in L$.]  In this case $u=x_{i', \ell}$ for some $i' \neq i$, and we distinguish the following sub-cases:
        \begin{itemize}
            \item If $v = y'_{i'}$ then  $\widetilde{P} = \langle x_{i', \ell}, y_{i'}, y'_{i'} \rangle$.
            
            \item If $v = y'_{i''}$ with $i'' \neq i'$, then  $\widetilde{P} = \langle x_{i', \ell}, y_{i'}, y_{i''}, y'_{i''} \rangle$.
            
            \item If $v = z'_{j'}$ and $x_{i'} \in S_{j'}$, then pick $\widetilde{P} = \langle x_{i', \ell}, S_{j', \ell}, z_{j'}, z'_{j'}  \rangle$.
            
            \item If $v = z'_{j'}$, $x_{i'} \not\in S_{j'}$, and there exists some $S_{j''}$ for which $S_{j'} \cap S_{j''} \neq \emptyset$, then let $x_{i''} \in S_{j'} \cap S_{j''}$ and choose $\widetilde{P} = \langle x_{i', \ell}, S_{j'', \ell}, x_{i'', \ell}, S_{j', \ell}, z_{j'}, z'_{j'}  \rangle$.
            
             \item If $v = z'_{j'}$, $x_{i'} \not\in S_{j'}$ and there is no $S_{j''}$ for which $S_{j'} \cap S_{j''} \neq \emptyset$, then let $x_{i''} \in S_{j'} \setminus \{x_i\}$. We choose $\widetilde{P} = \langle x_{i', \ell}, y_{i'}, y_{i''}, x_{i'', \ell}, S_{j', \ell}, z_{j'}, z'_{j'} \rangle$.
        \end{itemize}

        \item[Case 5: $u$ is a set-vertex in $H_\ell$ and $v \in L$.] In this case $u=S_{j', \ell}$ for some $j'$ such that $x_{i, \ell} \not\in S_{j'}$, and we distinguish the following sub-cases:
        \begin{itemize}
            \item If $v = y'_{i'}$ with $x_{i'} \in S_{j'}$, then $i' \neq i$ and we choose $\widetilde{P} = \langle  S_{j', \ell}, x_{i', \ell}, y_{i'}, y'_{i'} \rangle$.
            
            \item If $v = y'_{i'}$ with $x_{i'} \not\in S_{j'}$, then let $x_{i''} \in S_{j'}$ and choose  $\widetilde{P} = \langle  S_{j', \ell}, x_{i'', \ell}, y_{i''}, y_{i'}, y'_{i'} \rangle$.

            \item If $v = z'_{j'}$, then  $\widetilde{P} = \langle S_{j', \ell}, z_{j'}, z'_{j'} \rangle$.

            \item If $v = z'_{j''}$ with $j'' \neq j'$ and $S_{j'} \cap S_{j''} \neq \emptyset$, then let $x_{i'} \in S_{j'} \cap S_{j''}$ and choose  $\widetilde{P} = \langle S_{j', \ell}, x_{i', \ell}, S_{j'', \ell}, z_{j''}, z'_{j''} \rangle$.

            \item If $v = z'_{j''}$ with $j'' \neq j'$,  $S_{j'} \cap S_{j''} = \emptyset$, and there exists some $S_{j'''}$ such that $S_{j'} \cap S_{j'''} \neq \emptyset$ and $S_{j''} \cap S_{j'''} \neq \emptyset$, then let $x_{i'} \in S_{j'} \cap S_{j'''}$ and $x_{i''} \in S_{j''} \cap S_{j'''}$. We choose $\widetilde{P} = \langle  S_{j', \ell}, x_{i', \ell}, S_{j''', \ell}, x_{i'', \ell}, S_{j'', \ell}, z_{j''}, z'_{j''} \rangle$.

            \item If $v = z'_{j''}$ with $j'' \neq j'$,  $S_{j'} \cap S_{j''} = \emptyset$, and there is no  $S_{j'''}$ such that $S_{j'} \cap S_{j'''} \neq \emptyset$ and $S_{j''} \cap S_{j'''} \neq \emptyset$, then let $x_{i'} \in S_{j'}$ and $x_{i''} \in S_{j''} \setminus \{x_i\}$ and  choose  $\widetilde{P} = \langle  S_{j', \ell}, x_{i', \ell},  y_{i'}, y_{i''}, x_{i'', \ell}, S_{j'', \ell}, z_{j''}, z'_{j''} \rangle$.
        \end{itemize}
    \end{description}
\end{proof}

\begin{lemma}\label{lemma:meg_set_contains_set_cover}
    Given a MEG-set $M'$ of $G$, we can compute in polynomial time a MEG-set $M$ of $G$ such that $|M|\leq |M'|$ and, for every $\ell=1,\ldots,k$, the set $\mathcal{S}_\ell=\{S_j \in \mathcal{S} \mid S_{j,\ell} \in M\}$ is a set cover of $\mathcal{I}$. 
\end{lemma}
\begin{proof}
    Let $M''$ be a minimal MEG-set of $G$ that is obtained from $M'$ by possibly discarding some of the vertices. Clearly $|M''|\leq |M'|$ and $M''$ can be computed in polynomial time. Moreover, by \Cref{lemma:what_monitors_minimal_meg_set}, for every $i=1,\ldots,\eta$ and every $\ell=1,\ldots,k$, $M''$ contains $x_{i,\ell}$ or some $S_{j,\ell}$ such that $S_j$ covers $x_i$. We compute $M$ from $M''$ by replacing each $x_{i,\ell} \in M''$ with $S_{j,\ell}$, where $S_j \in \mathcal{S}$ is any set that covers $x_i$. As a consequence, for every $\ell=1,\ldots, k$, the set $\mathcal{S}_\ell=\{S_j \in \mathcal{S} \mid S_{j,\ell} \in M\}$ is a set cover of $\mathcal{I}$.
    Moreover, since $M''$ contains all vertices in $L$ by \Cref{lemma:leaf}, so does $M$. Then, \Cref{lemma:set_cover_induces_meg_set} implies that $M$ is a MEG-set of $G$.
\end{proof}

\begin{lemma}
\label{lemma:inapx}
Let $\varepsilon > 0$ be a constant of choice.
Any polynomial-time $(\alpha \ln n)$-approximation algorithm for the minimum MEG-set problem, where $\alpha>0$ is a constant, implies the existence of a polynomial-time $((2\alpha+\varepsilon) \ln N)$-approximation algorithm for \textsc{Set Cover}.
\end{lemma}
\begin{proof}

Given an instance $\mathcal{I}=\langle X, \mathcal{S}\rangle$ of \textsc{Set Cover} and let $h^*$ be the size of an optimal set cover of $\mathcal{I}$.
In the rest of the proof we assume w.l.o.g.\ that $N \ge 4$ and $h^* \ge \frac{4 \alpha}{\varepsilon}$.
Indeed, if any of the above two conditions does not hold, we can solve $\mathcal{I}$ in constant time.

We now construct the graph $G$ with $n=  (k+2) N \le N^2$ vertices by making $k = N-2$ copies of $H$.
Next, we run the $(\alpha \ln n)$-approximation algorithm to compute a MEG-set $M'$ of $G$, and we use \Cref{lemma:meg_set_contains_set_cover} to find a MEG-set $M$ with $|M| \le |M'|$ that contains $k$ set covers $\mathcal{S}_1,\ldots,\mathcal{S}_k$ in polynomial time. Among these $k$ set covers, we output one $\mathcal{S}'$ of minimum size. 

To analyze the approximation ratio of the above algorithm, let $M^*$ be an optimal MEG-set of $G$. 
\Cref{lemma:set_cover_induces_meg_set} ensures that $|M^*| \leq |L|+kh^* = N + kh^*$, and hence
\[
    |M| \le |M'| \le \alpha (N+kh^*) \ln n = \alpha (N+kh^*) \ln N^2=
     2 \alpha k h^* \ln N + 2 \alpha N \ln N.
\]
Therefore we have:
\begin{align*}
    |\mathcal{S}'| & \le \frac{|M|}{k} 
    \le 2 \alpha h^* \ln N + \frac{2 \alpha  N\ln N}{k}  \\
    & \le  2 \alpha h^* \ln N + 4 \alpha  \ln N 
    =\left(2 \alpha + \frac{4\alpha}{h^*} \right) h^* \ln N
    \le (2\alpha + \varepsilon) h^* \ln N.  \qedhere
\end{align*}
\end{proof}

 Let $\gamma$ be any positive constant. Since \textsc{Set Cover} cannot be approximated in polynomial time within a factor of $(1 - \gamma) \ln |\mathcal{I}|$, unless $\p = \np$ \cite{DinurS14}, and since an invocation of \Cref{lemma:inapx} with  $\alpha = \frac{1}{2} - \gamma$ and $\varepsilon = \gamma$
 shows that any polynomial-time $( (\frac{1}{2} - \gamma) \ln n)$-approximation algorithm for the minimum MEG-set problem can be turned into a polynomial-time approximation algorithm for \textsc{Set Cover} with an approximation ratio of $( \frac{1}{2} - \gamma) \ln N \le ( \frac{1}{2} - \gamma) \ln |\mathcal{I}|$, we have:
 
\begin{theorem}
The minimum MEG-set problem cannot be approximated in polynomial time within a factor of $c \ln n$, for any constant $c < \frac{1}{2}$, unless \p = \np.
\end{theorem}

\bibliographystyle{unsrt} 
\bibliography{references}

\end{document}